%% file: main.tex
\documentclass{l4dc2027}

\usepackage{times}
\usepackage{amsmath}
\usepackage{booktabs}
\usepackage{float}
\usepackage{graphicx}

\newif\ifshowrevisions
\newcommand{\rev}[1]{{\ifshowrevisions\color{blue}\fi #1}}
\newenvironment{revision}{\begingroup\ifshowrevisions\color{blue}\fi}{\endgroup}
\newcommand{\corr}[1]{{\ifshowrevisions\color{red}\fi #1}}
\newenvironment{correction}{\begingroup\ifshowrevisions\color{red}\fi}{\endgroup}
\title[EDORL]{Expert-Guided Online Reinforcement Learning for Stable and Sample-Efficient Control of \rev{Dynamical Systems}}
\DeclareMathOperator{\Tr}{Tr}

\author{%
 \Name{Supratim Manna} \Email{supratim.manna@liu.se}\\
 \addr Electrical Engineering Department, Linköping University, 581 83, Sweden
 \AND
 \Name{Arunava Naha} \Email{arunava.naha@liu.se}\\
 \addr Electrical Engineering Department, Linköping University, 581 83, Sweden%
 \AND
 \Name{Farnaz Adib Yaghmaie} \Email{farnaz.adib.yaghmaie@liu.se}\\
 \addr Electrical Engineering Department, Linköping University, 581 83, Sweden%
}

\begin{document}
\raggedbottom

\newtheorem{Theo}{Theorem}
\newtheorem{Lem}{Lemma}
\newtheorem{Alg}{Algorithm}
\newtheorem{Assum}{Assumption}
\newtheorem{Rem}{Remark}
\newtheorem{Cor}{Corollary}
\newtheorem{Def}{Definition}
\newtheorem{Prob}{Problem}
\newtheorem{Fact}{Fact}
\newtheorem{Prop}{Proposition}
\newtheorem{Exam}{Example}

\maketitle

\begin{abstract}%
Online reinforcement learning (RL) often requires extensive interactions with the environment, which makes the training costly. Additionally, exploration can destabilise the dynamical system, which may be unsafe. On the other hand, offline learning from expert demonstrations avoids unsafe exploration but is limited by the quality and coverage of the offline dataset. In this paper, we propose a model-free RL approach that combines expert demonstrations with limited online interactions. \rev{Specifically, we optimise a policy objective consisting of an online performance term and a behavioural cloning regulariser computed from offline expert demonstrations.} The regulariser encourages the learned policy to stay close to expert behaviour, reducing the need to explore unfamiliar actions. \rev{Expert demonstrations guide policy learning and reduce the need for many online interactions and exploration.} \rev{For stochastic linear time-invariant (LTI) systems, we establish conditions for linear convergence of exact policy gradient updates and for all intermediate policies to remain stabilising.} Although our theoretical analysis focuses on linear systems, we also evaluate the proposed approach on nonlinear systems to demonstrate its effectiveness beyond the linear setting. \rev{Numerical results show that the proposed method reaches comparable performance in fewer gradient steps and with less number of samples than a purely online RL baseline.}
\end{abstract}


\input{sections/introduction}

\input{sections/related_work}
\input{sections/notation_flow_q}

\input{sections/algorithm}

\input{sections/theo_analysis}
\input{sections/experimental_analysis}

\FloatBarrier
\begin{samepage}
\input{sections/conclusion}
\par
\end{samepage}


\bibliography{references}
\clearpage 
\appendix
\input{sections/appendix}
\end{document}

%% file: sections/introduction.tex
\section{Introduction}

\rev{Online reinforcement learning (RL) learns through live interaction with the environment, but it is sample inefficient, computationally demanding, and can drive the system into unsafe regions while exploring \rev{\citep{ghasemi2024comprehensive, arulkumaran2017deep}}. In many problems, expert trajectories are available and can be used through imitation learning, such as behavioural cloning \rev{\citep{torabi2018behavioral}} or inverse RL \rev{\citep{ng2000algorithms}}, typically in an offline setup that trains on pre-collected data without further interaction \rev{\citep{levine2020offline}}. These methods have shown promising results in control, but they suffer from distribution shift, the mismatch between the states seen during training and those encountered at test time. They also assume that the expert data is of good quality, which is not always the case, and as datasets grow larger and more complex, purely offline learning becomes harder.}

\rev{One promising direction is to use expert demonstrations (offline data) together with online data from real-time exploration. This can mitigate the sample inefficiency and unsafe exploration of online RL, as well as the distribution shift of offline RL. Despite promising empirical results, the convergence and stability of such methods for closed-loop dynamical systems have not been well studied.}

In this paper, we ask whether a policy can be trained jointly on expert demonstrations and online data from live interaction, with a behaviour cloning (BC) term on the expert data as a regulariser, and without pre-training on offline expert data or hybrid $Q$-learning. We propose \textit{Expert-Driven Online RL} (EDORL), which uses a simple sampling scheme for the expert data and the BC technique of Flow-Q learning \rev{\citep{park2025flow}}. Furthermore, we show its efficacy on a linearised inverted pendulum and on several nonlinear environments from MuJoCo \rev{\citep{todorov2012mujoco}}. We also study the theoretical properties of EDORL for Linear Quadratic Regulator (LQR) problems. Our analysis focuses on the convergence and stability of a policy learning algorithm that optimises the average expected cost regularised by BC from expert trajectories. \rev{Our contributions are as follows.}
\begin{enumerate}
\item \rev{We propose a model-free RL approach that combines online policy optimisation with a behavioural cloning regulariser based on offline expert demonstrations.} \rev{The regulariser guides policy learning towards expert behaviour.}
\item \rev{For stochastic linear time-invariant (LTI) systems, we establish conditions for gradient dominance and smoothness, yielding linear convergence of exact policy gradient updates. Additionally, we show that under the stated assumptions and step-size conditions, all intermediate policies remain stabilising for the LTI system.} \corr{We also study how the regularisation strength $\alpha$ affects convergence and the trade-off between expert matching and control performance.}
\item We evaluate the proposed method on both linear and nonlinear dynamical systems. \rev{Numerical results show that the proposed approach reaches comparable performance in fewer gradient steps than its purely online counterpart.}
\end{enumerate}

\rev{Section~\ref{sec:related_work} reviews related work, and Section~\ref{sec:back} introduces the preliminaries. Section~\ref{sec:edorl} presents EDORL, followed by the LQR analysis in Section~\ref{Sec:theoretical_analysis}. Section~\ref{Sec:experimental_analysis} reports the experiments, and Section~\ref{sec:conclusion} concludes.}

%% file: sections/related_work.tex
\section{\rev{Related Work}}
\label{sec:related_work}
\begin{revision}
Actor-critic (AC) methods \rev{\citep{konda2000actor, lillicrap2015continuous, schulman2015trust, mnih2016asynchronous}} are widely used in RL, with applications in robotics \rev{\citep{kober2013reinforcement}}, game playing \rev{\citep{silver2016mastering}}, and autonomous systems \rev{\citep{kiran2021deep}}. They combine the strengths of temporal difference (TD), gradient TD, or similar techniques for value function learning and policy gradient (PG) methods for policy optimisation. Furthermore, AC-based methods can handle continuous state and action spaces as well. The actor selects actions through a policy network, while the critic evaluates them by estimating the value function or the action-value ($Q$) function. Because the policy is parameterised directly, AC methods avoid the costly value-function optimisation of conventional TD-based approaches \rev{\citep{sutton1998reinforcement}}, which makes them a good fit for control problems with continuous action spaces. Still, they can suffer from stability and convergence issues when applied to closed-loop dynamical systems \rev{\citep{henderson2018deep}}.
\end{revision}

\begin{revision}
\paragraph{Combining offline and online learning.} Approaches that combine offline and online learning often first pre-train the policy on expert data and then fine-tune it with online RL. In contrast, \citet{ball2023efficient, song2022hybrid} mix the expert data with the online data during training, using hybrid $Q$-learning and a policy entropy term following \rev{\citep{haarnoja2018soft}}. Despite their promising empirical results, the convergence and stability of such methods for closed-loop dynamical systems have not been well studied.
\end{revision}

\begin{revision}
\paragraph{Flow-based policies.} Flow $Q$-learning \rev{\citep{park2025flow}} is a recent AC approach that brings expert trajectories into the RL framework and learns expressive policies with flow-based generative models. EDORL uses this behavioural cloning approach to regularise online policy learning, while keeping the critic updates on the online replay buffer. Note that a different BC-based regularisation could be used, but the flow-based policy allows for more flexible and expressive modelling of the action distribution, which can be particularly beneficial in high-dimensional or complex control tasks.
\end{revision}

%% file: sections/notation_flow_q.tex
\section{\rev{Preliminaries}}
\label{sec:back}
\subsection{Notation}
\noindent
Let $\mathbb{R}^{m \times n}$ denote the set of all real-valued matrices of dimension $m \times n$. A symmetric positive definite (positive semi-definite) matrix $P$ of size $n \times n$ is denoted by $P \succ 0$ ($P \succeq 0$). The Kronecker product of two matrices $A$ and $B$ is written as $A \otimes B$. For a matrix $A$, the Frobenius norm is defined as $\|A\|_F = \sqrt{\sum_{i,j} |a_{ij}|^2}$, and the spectral norm is given by $\|A\|_2$, which corresponds to the largest singular value of $A$. Let $\rho(A)$ denote the spectral radius of $A$; i.e. $\rho(A) = \max \{ |\lambda| \in \mathbb{R} : \lambda \text{ is an eigenvalue of } A \}$. The vectorization of a matrix $A \in \mathbb{R}^{m \times n}$ is denoted by $\text{vec}(A) \in \mathbb{R}^{mn}$, which stacks the columns of $A$ into a single column vector. The trace of a square matrix $A$ is denoted by $\Tr(A)$, which is the sum of its diagonal elements. Throughout this paper, we use the subscript $k$ to refer to the time step for the dynamical system and the superscript $t$ to refer to the iteration index in the RL algorithm.

\subsection{\rev{Markov Decision Processes}}
We consider a Markov Decision Process (MDP) defined by the tuple $(\mathcal{X}, \mathcal{U}, P, c, \nu_0)$, where $\mathcal{X} \subseteq \mathbb{R}^{n}$ and $\mathcal{U} \subseteq \mathbb{R}^{m}$ denote the state and action spaces, $P: \mathcal{X} \times \mathcal{U} \times \mathcal{X} \rightarrow [0,1]$ is the transition probability function such that $P(x' \mid x, u)$ denotes the probability of transitioning to state $x'$ from state $x$ under action $u$, $c: \mathcal{X} \times \mathcal{U} \rightarrow \mathbb{R}$ is the stage cost, and $\nu_0$ is the initial state distribution. We reserve the symbol $\Sigma$ for covariance matrices throughout. The aim is to learn a parameterised policy $\pi_{\varphi}(.|x): \mathcal{X} \rightarrow \mathcal{U}$ with parameter $\varphi$ that minimises the average cost performance index as
{\setlength{\belowdisplayskip}{2mm} \setlength{\abovedisplayskip}{2mm}
\begin{align}
    J(\pi_{\varphi}) = \lim_{T \to \infty} \frac{1}{T} \mathbb{E}^{\pi_{\varphi}}_{x_0 \sim \nu_0} \left[ \sum_{k=0}^{T-1} c(x_k, u_k) \right], \label{eq:average_cost1}
\end{align}
}
where $\mathbb{E}^{\pi_{\varphi}}$ denotes the expectation over the trajectory $\{(x_k,u_k)\}$ generated by the transition kernel $P$ and by sampling the actions $u_k \sim \pi_{\varphi}(\cdot|x_k)$. To learn the average cost, one needs to assume ergodicity of the MDP under the policy $\pi_{\varphi}$ \rev{\citep{puterman2014markov}}. Under the ergodicity assumption the limit exists and is independent of the initial distribution $\nu_0$. Let $\rho_{\varphi}$ denote the stationary state distribution induced by $\pi_{\varphi}$, which exists by ergodicity and satisfies
\begin{align}
J(\pi_{\varphi})= \mathbb{E}_{x\sim \rho_{\varphi}, u \sim \pi_{\varphi}(\cdot|x)}[c(x,u)].
\label{eq:average_cost}
\end{align}
Here we simplified the notation by omitting the explicit dependence on the policy in the stationary distribution. We emphasise that $\rho_{\varphi}$ is induced by the policy $\pi_{\varphi}$ itself, and therefore changes as the policy is updated. The action-value function for the policy $\pi_{\varphi}$ is defined as
{\setlength{\belowdisplayskip}{2mm} \setlength{\abovedisplayskip}{2mm}
\begin{align}
    Q^{\pi_{\varphi}} (x,u)=\lim_{T \to \infty}\mathbb{E}_{\pi_{\varphi}}\left[\sum_{k=0}^{T-1}\left(c(x_k, u_k)-  J(\pi_{\varphi})\right)\Big|x_0=x,u_0=u\right].
\end{align}
}
\subsection{\rev{Flow-Based Behavioural Cloning}}
\label{subsec:flow_bc}
Behaviour cloning \citep{florence2022implicit} learns to imitate actions in an offline dataset and is often used to regularise policy learning, balancing cost minimisation with proximity to the demonstrated behaviour. Flow matching
\citep{lipman2024flow, liu2022flow} models complex data distributions by matching the flow of points
from a known input distribution to an unknown target. \rev{Flow $Q$-learning} (\rev{\citep{park2025flow}}) regularises the AC structure with behavioural cloning from expert trajectories. It learns a behaviour cloning \rev{policy} $\pi_{\theta_b}$ which is trained by flow matching. Intuitively, the behavioural cloning policy $\pi_{\theta_b}$ maps the noise $z \sim \mathcal{N}(0,W_z)$ to the action $a$ via an ordinary differential equation with velocity $v_{\theta_b}(t, x, z)$. The flow velocity is related to the flow through $\frac{d}{dt}\psi_{\theta_b}(t,x,z)=v_{\theta_b}(t,x,z)$. Now a parameterised vector field $v_{\theta_b} : [0,1]\times\mathcal{X}\times\mathbb{R}^{n_u}\to\mathbb{R}^{n_u}$ is trained by the following flow matching loss,
{\setlength{\belowdisplayskip}{2mm} \setlength{\abovedisplayskip}{2mm}
\begin{equation}
    \mathcal{L}_{\mathrm{Flow}}(\theta_b)=\mathbb{E}_{x, a^1(=u) \sim \mathcal{D};\; a^{0} (=z) \sim \mathcal{N}(0, I);\; t \sim \mathrm{Unif}[0,1]}
    \bigl[\|v_{\theta_b}(t, x, a^t) - (a^1 - a^0)\|_{2}^{2}\bigr],
    \label{eq:flow_matching}
\end{equation}
}
\rev{where} $a^t = (1-t)a^{0} + ta^{1}$ and $t \in (0,1)$. Then, the behavioural cloning policy $\pi_{\varphi_b}$ is given by $\pi_{\varphi_b}(x,z)=\psi_{\theta_b}(1, x, z)$, which is used to evaluate the behaviour cloning term for regularisation in EDORL. Further details on the flow matching can be found in \rev{\citep{park2025flow}}.

%% file: sections/algorithm.tex
\section{Proposed Expert-Driven Online RL (EDORL)}
\label{sec:edorl}
We consider the standard RL setting, with the addition of an offline dataset of expert demonstrations. The offline data may come from human demonstrators or from previously trained policies, whether optimal or sub-optimal. The algorithm itself does not require a particular quantity or quality of expert data, \rev{although its performance depends} on how well the data cover the states the learned policy visits.

\rev{We train the policy jointly on online interaction data and previously stored expert demonstrations, without offline pre-training of the actor or hybrid $Q$-learning. The behavioural cloning policy is trained separately on the expert dataset. We describe the regularised objective and the actor--critic updates below. Algorithm~\ref{alg:edorl} in Appendix~\ref{app:edorl_algo} gives the full procedure.}

\subsection{Problem Setup}
\label{subsec:problem_setup}
We consider a system with state $x_t \in \mathbb{R}^n$ and action $u_t \in \mathbb{R}^m$, evolving as $x_{t+1} \sim P(\cdot \mid x_t, u_t)$ under a stage cost $c(x_t, u_t)$. The learner has access to two data sources, a fixed expert dataset $\mathcal{D}_{\mathrm{e}} = \{(x_i, u_i, c_i, x_i')\}_{i=1}^{N}$, collected beforehand and never updated, and an online replay buffer $\mathcal{D}_{\mathrm{o}}$, which starts empty and grows as the current policy $\pi_\varphi$ interacts with the system. \rev{Our goal is to learn a policy} $\pi_\varphi$ \rev{that balances the expected long-run average expected cost $J(\varphi)$, as in (\ref{eq:average_cost}), estimated from online interaction, with a behavioural cloning regulariser $\mathcal{R}(\varphi)$ evaluated on the fixed expert dataset. Therefore, the training objective is}
{ \setlength{\belowdisplayskip}{2mm} \setlength{\abovedisplayskip}{2mm}
\begin{equation}
    \mathcal{L}(\varphi) = J(\varphi) + \frac{\alpha}{2} \mathcal{R}(\varphi),
    \label{eq:joint_cost_function}
\end{equation}
}
where $\alpha\geq 0$ controls the regularisation strength. Here $\mathcal{R}(\varphi)$ is the behavioural cloning regulariser, defined as
\setlength{\belowdisplayskip}{2mm} \setlength{\abovedisplayskip}{2mm}
\begin{equation}
\mathcal{R}(\varphi) := \mathbb{E}_{ x \sim \mathcal{D}_{\mathrm{e}},\, \eta \sim \mathcal{N}(0,\Sigma_{\eta})}
\Vert \pi_{\varphi}(x, \eta) - \pi_{\varphi_b}(x, \eta)\Vert_2^2
\label{eq:bc_regularizer}
\end{equation}
\rev{We first train $\pi_{\varphi_b}$ by flow matching, using \eqref{eq:flow_matching}, and then keep it fixed to evaluate the BC regulariser.}

\subsection{\rev{Actor--Critic Training}}
\label{subsec:sampling_technique}
At each gradient step, we build a training batch of size $B$ from two sources: the fixed expert dataset $\mathcal{D}_{\mathrm{e}}$ and the online replay buffer $\mathcal{D}_{\mathrm{o}}$. Given a mixing ratio $\mathcal{\zeta} \in [0,1]$, we draw $B_{\mathrm{e}}=\lfloor \mathcal{\zeta} B \rfloor$ transitions uniformly from $\mathcal{D}_{\mathrm{e}}$ and the remaining $B_{\mathrm{o}}=B - \lfloor \mathcal{\zeta} B \rfloor$ uniformly from $\mathcal{D}_{\mathrm{o}}$. \rev{The critic is trained only on $B_{\mathrm{o}}$, using transitions collected during online learning.} In the actor update, the two parts play different roles. The behaviour cloning term $\mathcal{R}(\varphi)$ is evaluated only on the expert part $B_{\mathrm{e}}$, while the performance term $J(\varphi)$ is evaluated on the online part $B_{\mathrm{o}}$. The expert dataset is never modified, and the online buffer is the only data that changes during training. \rev{After training the behavioural cloning policy $\pi_{\varphi_b}$, both data sources are used from the first actor update. Here, the expert data anchors the policy, providing a stable reference for imitation, while online data supports performance improvement by allowing the policy to adapt to new experiences. Expert data generation process followed in our study is described in Appendix~\ref{subsec:expert_data_creation}.}

%% file: sections/theo_analysis.tex
\section{Theoretical Analysis}
\label{Sec:theoretical_analysis}

\begin{revision}
We analyse the convergence properties of the exact policy gradient updates for the LQR setup of EDORL. We adopt the average-cost formulation (\ref{eq:average_cost}) because closed-loop control problems are naturally infinite-horizon and concerned with long-run steady-state performance. Unlike discounted formulations, which reduce future costs, the average cost better captures the asymptotic behaviour and regulation quality of the stabilised system. Moreover, the average cost is coercive, so it acts as a natural barrier that helps preserve stability during learning, as discussed in the following sections. Finally, working with the average expected cost instead of the $Q$-function used in \rev{\citep{park2025flow}} lets us apply the policy gradient framework \rev{\citep{silver2014deterministic}} to analyse the convergence of the learned policy.
\end{revision}

\subsection{LQR Problem}
\label{subsec:lqr}
We consider the following LTI system with unknown parameters $\theta:=(A,B)$ and fully measured state,
{\setlength{\belowdisplayskip}{2pt} \setlength{\abovedisplayskip}{2pt}
\begin{align}
    x_{k+1}= A x_{k} +B u_k +w_k,
    \label{eq:lin:dynamics}
\end{align}
}
where $x_k \in \mathbb{R}^{n}, u_k \in \mathbb{R}^{m}$ are the state and the control input and $w_k \in \mathbb{R}^{n}\sim \mathcal{N}(0, W_w)$. We consider a quadratic stage cost as follows,
{\setlength{\belowdisplayskip}{2pt} \setlength{\abovedisplayskip}{2pt}
\begin{align}
    c(x_k,u_k) = \frac{1}{2} x_k^{\top} R_x x_k +\frac{1}{2} u_k^{\top} R_u u_k,
    \label{Eq:Quad:cost}
\end{align}
}
where $R_x \succeq 0 $ and $R_u \succ 0$. For linear systems and quadratic costs, we are interested in learning a linear state feedback controller $u_k = \pi_{K} (x_k,z_k) = K x_k+z_k, \: z_k \sim \mathcal{N}(0,W_z),\: K \in \mathcal{K}$ where $\mathcal{K}$ is the feasible set of the stabilising controller gains, i.e., $\mathcal{K}=\{K \in \mathbb{R}^{m \times n} \mid \rho(A+BK) <1 \}$. 
\begin{correction}
\begin{Assum}[LQR regularity]\label{Assum:lqr_regularity}
For the theoretical results, assume $\mathcal{K}\ne\varnothing$, $R_x\succ0$, $R_u\succ0$, and $\Omega:=BW_zB^{\top}+W_w\succ0$. The covariance matrices are fixed, the noise sequences $\{z_k\}$ and $\{w_k\}$ are mutually independent and i.i.d., and both are independent of $x_0$, which has finite second moment. We allow $W_z\succeq0$.
\end{Assum}
\end{correction}
The average expected cost associated with the policy $\pi_K$ is given as
{\setlength{\belowdisplayskip}{2pt} \setlength{\abovedisplayskip}{2pt}
\begin{align}
    \corr{J(K) =\lim_{T \to \infty}\frac{1}{T}\mathbb{E}^{\pi_K}_{x_0\sim\nu_0}\left[\sum_{k=0}^{T-1} c(x_k,u_k) \right]
    = \mathbb{E}_{x\sim\rho_K,\,u\sim\pi_K(\cdot|x)}[c(x,u)]}
    \label{eq:Ave:Quad:cost}
\end{align}
}
with $c(x_k,u_k)$ given in \eqref{Eq:Quad:cost}. Since $K\in\mathcal{K}$, the closed loop is ergodic and the limit exists and does not depend on $\nu_0$. As shown in Lemma \ref{Lem:sigma} (Appendix~\ref{app:proof_Lem_sigma}), the limit is characterized by the stationary state covariance $\Sigma_K$ induced by the policy $\pi_K$ itself.
Optimizing \eqref{eq:Ave:Quad:cost} for the linear dynamical system in \eqref{eq:lin:dynamics} defines an LQR problem. \rev{In the LQR setting, the EDORL policy objective is}
{\setlength{\belowdisplayskip}{2pt} \setlength{\abovedisplayskip}{2pt}
\begin{align}
    l_{\varphi}^{\text{lin}}(K) = J(K) + \frac{\alpha}{2} \mathcal{R}(K)
    \label{eq:edorl_lqr}
\end{align}
}
where, $\mathcal{R}(K)$ comes from Eq.~\eqref{eq:bc_regularizer}. \rev{An important observation is that} the two terms of the policy loss in Eq.~\eqref{eq:edorl_lqr} are defined with respect to two \emph{different} distributions, and conflating them invalidates the analysis (See Appendix~\ref{app:data_buffer} for more details).
\subsection{\rev{EDORL Policy Loss and Expert Data Assumptions}}
\label{subsec:lqr_policy_loss}
\begin{revision}
\begin{correction}
\begin{Assum}[Fixed expert policy and data coverage]\label{Assum:data}
The expert dataset $\mathcal{D}_{\mathrm{e}}$ and the fitted behaviour policy $\pi_{\varphi_b}$ are fixed during policy optimisation. Expectations over $x\sim\mathcal{D}_{\mathrm{e}}$ use its empirical state distribution, independently of $z\sim\mathcal{N}(0,W_z)$. Assume
\begin{align}
G_{\mathcal{D}_{\mathrm{e}}}:=\mathbb{E}_{x\sim\mathcal{D}_{\mathrm{e}}}[xx^{\top}]
\succeq\sigma I_n,\qquad \sigma>0, \label{eq:GD}
\end{align}
and $\mathbb{E}_{x,z}\|\pi_{\varphi_b}(x,z)\|_2^2<\infty$.
\end{Assum}
\end{correction}

Since $\mathcal{D}_\mathrm{e}$ is fixed, $G_{\mathcal{D}_{\mathrm{e}}}$ is a \emph{constant} matrix that does not depend on $K$. We further define the (constant) expert cross-moment as,
{\setlength{\belowdisplayskip}{5pt} \setlength{\abovedisplayskip}{2pt}
\begin{align}
    \corr{M:=\mathbb{E}_{\substack{x\sim\mathcal{D}_{\mathrm{e}}\\z\sim\mathcal{N}(0,W_z)}}[\pi_{\varphi_b}(x,z)x^{\top}].}
    \label{eq:M}
\end{align}
}
\end{revision}
The final policy loss for EDORL in the LQR problem is given by (from Eq.~\eqref{eq:bc_regularizer} and Eq.~\eqref{eq:edorl_lqr})
{\setlength{\belowdisplayskip}{2pt} \setlength{\abovedisplayskip}{7pt}
\begin{align}
    l_{\varphi}^{\text{lin}}(K)=& \underbrace{J(K)}_{\text{closed-loop average cost}} \notag\\
    &+\underbrace{\frac{\alpha}{2}\mathbb{E}_{\substack{x\sim \mathcal{D}_{\mathrm{e}}\\z \sim \mathcal{N}(0,W_z)}}\left[\Vert Kx+z-\pi_{\varphi_b}(x, z)\Vert_2^2\right]}_{=:L_{\text{Distill}}(K)}.
    \label{eq:loss_one_step_LQ}
\end{align}
}
\rev{Note that} the expectation over \corr{$x\sim\mathcal{D}_{\mathrm{e}}$} applies \emph{only} to the behavioural cloning term. The term $J(K)$ \eqref{eq:Ave:Quad:cost} is evaluated under the stationary distribution $\rho_K$ of the current policy, and by Lemma \ref{Lem:avecost} (Appendix~\ref{app:proof_Lem_sigma}) it equals $\tfrac{1}{2}\Tr((R_x+K^{\top}R_uK)\Sigma_K)+\tfrac{1}{2}\Tr(R_uW_z)$. \rev{The controller gain} $K$ is updated via gradient descent over $ l_{\varphi}^{\text{lin}}$
{\setlength{\belowdisplayskip}{2pt} \setlength{\abovedisplayskip}{6pt}
\begin{align}
    K^{t+1} = K^{t} -\eta \nabla_{K} l_{\varphi}^{\text{lin}}(K^t).
    \label{Eq:gradient:K}
\end{align}
}
\begin{correction}
\begin{Lem}[Coercivity]\label{lem:coercive}
Under Assumptions~\ref{Assum:lqr_regularity} and~\ref{Assum:data}, for every fixed $\alpha\geq0$, the policy loss in \eqref{eq:loss_one_step_LQ} is coercive on $\mathcal{K}$, i.e., it grows unbounded as $\|K\|_F\to\infty$ or as $K$ approaches the boundary of the stabilising set, and has compact sublevel sets.
\end{Lem}
\end{correction}
\begin{proof}
See Appendix~\ref{app:proof_lem_coercive}.
\end{proof}
\subsection{Convergence Analysis}
\label{subsec:convergence_analysis}
\begin{correction}
Fix $\alpha>0$ and $K^0\in\mathcal{K}$, and set $\gamma_0=l_{\varphi}^{\mathrm{lin}}(K^0)$ and $\mathcal{K}_{\gamma}:=\{K\in\mathcal{K}:l_{\varphi}^{\mathrm{lin}}(K)\leq\gamma\}$. The sublevel set $\mathcal{K}_{\gamma_0}$ is compact by Lemma~\ref{lem:coercive}. The objective is analytic on $\mathcal{K}$, but need not be globally smooth. The gradient is given in Lemma~\ref{Lem:gradient}, and the optimisation definitions appear in Appendix~\ref{app:optimization_definitions}.

\begin{Assum}[Compact convex working region]\label{Assum:convex}
There exist a compact convex set $\mathcal{C}\subset\mathcal{K}$ and $r>0$ such that
\begin{align*}
\{K:\inf_{K'\in\mathcal{K}_{\gamma_0}}\|K-K'\|_F\leq r\}\subseteq\mathcal{C}.
\end{align*}
\end{Assum}
This is an additional geometric assumption and the compactness of a sublevel set alone does not imply it. Define the curvature bound on the entire working region,
\begin{align}
c_J(\mathcal{C}):=\max_{K\in\mathcal{C}}
\max\{0,-\lambda_{\min}(\nabla_K^2J(K))\}<\infty.\label{eq:cJ}
\end{align}
The Hessian is taken with respect to $\mathrm{vec}(K)$. The region, and hence this bound, may depend on $\alpha$ through $\mathcal{K}_{\gamma_0}$.

\begin{Theo}[Gradient dominance of the EDORL policy loss]
\label{Theo:EDORL_grad_dominence}\label{Theo:gradient_dominant}
Let Assumptions~\ref{Assum:lqr_regularity}, \ref{Assum:data}, and~\ref{Assum:convex} hold. If
\begin{align}
\alpha\sigma>c_J(\mathcal{C}),\label{eq:alpha_condition}
\end{align}
then $l_{\varphi}^{\mathrm{lin}}$ is $\mu$-strongly convex on $\mathcal{C}$, and
\begin{align}
\mu:=\alpha\sigma-c_J(\mathcal{C})>0\label{eq:mu}
\end{align}
where $\sigma$ is given in \eqref{eq:GD}. It has a unique global minimiser $K_\alpha^*$ over $\mathcal{K}$, which belongs to $\mathcal{K}_{\gamma_0}\subset\operatorname{int}\mathcal{C}$. For every $K\in\mathcal{C}$,
\begin{align}
l_{\varphi}^{\mathrm{lin}}(K)-l_{\varphi}^{\mathrm{lin}}(K_\alpha^*)
\leq\frac{\|\nabla_Kl_{\varphi}^{\mathrm{lin}}(K)\|_F^2}{2\mu}.
\label{eq:grad_dominant}
\end{align}
\end{Theo}
\end{correction}
\begin{proof}
See Appendix~\ref{app:proof_Theo_gradient_dominant}.
\end{proof}

\begin{correction}
\begin{Theo}[Stability and linear convergence]\label{Theo:EDORL_Stab}
Under the hypotheses of Theorem~\ref{Theo:EDORL_grad_dominence}, let $L$ be the smoothness constant on $\mathcal{C}$ specified in Lemma~\ref{Lem:Lsmooth}, and define
\begin{align}
\Gamma_0:=\max_{K\in\mathcal{K}_{\gamma_0}}
\|\nabla_Kl_{\varphi}^{\mathrm{lin}}(K)\|_F<\infty.\label{eq:Gamma0}
\end{align}
Choose a constant step size satisfying $0<\eta\leq1/L$ and $\eta\Gamma_0\leq r$. Then the exact gradient iterates in \eqref{Eq:gradient:K} satisfy, for all $t\geq0$,
\begin{enumerate}
\item[(i)] $l_{\varphi}^{\mathrm{lin}}(K^{t+1})\leq l_{\varphi}^{\mathrm{lin}}(K^t)\leq\gamma_0$ and $K^t\in\mathcal{K}_{\gamma_0}$;
\item[(ii)] $\rho(A+BK^t)<1$, so every intermediate policy is stabilising;
\item[(iii)] with $q:=1-\eta\mu\in[0,1)$,
\begin{align}
l_{\varphi}^{\mathrm{lin}}(K^t)-l_{\varphi}^{\mathrm{lin}}(K_\alpha^*)
\leq q^t\bigl(\gamma_0-l_{\varphi}^{\mathrm{lin}}(K_\alpha^*)\bigr),
\label{eq:linear_rate}
\end{align}
and $\|K^t-K_\alpha^*\|_F^2\leq(2/\mu)q^t(\gamma_0-l_{\varphi}^{\mathrm{lin}}(K_\alpha^*))$.
\end{enumerate}
The step-size condition also covers $\Gamma_0=0$, without division by zero.
\end{Theo}
\end{correction}
\begin{proof}
See Appendix~\ref{app:proof_Theo_Stab}.
\end{proof}

\begin{correction}
\begin{Rem}[Effect of the regularisation weight]\label{rem:alpha_effect}
On a fixed working region $\mathcal{C}$, increasing $\alpha$ increases the strong-convexity parameter $\mu=\alpha\sigma-c_J(\mathcal{C})$, but also the smoothness bound $L=L_J(\mathcal{C})+\alpha\|G_{\mathcal{D}_{\mathrm{e}}}\|_2$. Therefore, the convergence factor $(1-\eta\mu)$ depends on the admissible step size as well. Since $\mathcal{C}$, $r$, and $\Gamma_0$ may change with $\alpha$, a larger weight does not necessarily imply faster convergence.

The limiting gain $K_\alpha^*$ minimizes the regularized objective, rather than $J$ alone. If $K_{\mathrm{LQR}}$ minimizes $J$ over $\mathcal{K}$, optimality gives
\begin{align*}
0\leq J(K_\alpha^*)-J(K_{\mathrm{LQR}})
\leq\frac{\alpha}{2}\mathcal{R}(K_{\mathrm{LQR}}).
\end{align*}
Thus $\alpha$ balances expert matching against control performance. Stronger regularisation can bias the solution towards an imperfect expert, whereas a smaller weight reduces this cost-gap bound but may fail the sufficient curvature condition in \eqref{eq:alpha_condition}. Its choice should therefore account for expert quality and data coverage, as well as the convergence conditions.
\end{Rem}
\end{correction}

%% file: sections/experimental_analysis.tex
\section{\rev{Experiments}}
\label{Sec:experimental_analysis}

In this section, we present the empirical results of EDORL on several continuous control environments. We first apply it to an LQR setup, a \rev{linearized inverted pendulum} (see Appendix~\ref{app:linear_system_simulation} for the details), and then to nonlinear systems from MuJoCo \citep{todorov2012mujoco} available in Gymnasium \citep{towers2024gymnasium}. Specifically, we evaluate on environments, ranging from simple to complex systems, including \texttt{InvertedPendulum-v5}, \texttt{InvertedDoublePendulum-v5}, and \texttt{Hopper-v5}. Implementation details of the models are given in Appendix~\ref{app:implement_details}.
\subsection{\rev{Experimental Setup and Evaluation}}
\label{subsec:setting_experiment}

\rev{Our main metric is the number of gradient steps needed to reach the expert performance level. We also report evaluation reward as a function of training progress.} We compare three settings: (1) online RL, (2) EDORL (our approach), and (3) an online+offline baseline that uses both $\mathcal{D}_{\mathrm{o}}$ and $\mathcal{D}_{\mathrm{e}}$ for critic and $J(\varphi)$ evaluation unlike EDORL where only $\mathcal{D}_{\mathrm{o}}$ is used, as described in Subsection~\ref{subsec:sampling_technique}.
\begin{revision}
For each environment, we train each method with 5 different seeds. At each evaluation point, we run the learned policy for 20 episodes and report the mean reward, with shaded bands computed as $\text{mean} \pm 1.96\,(\text{std})$ across seeds. For a given environment, all methods use the same network architecture (a TD3-like architecture \rev{\citep{fujimoto2018addressing}}), optimizers, and evaluation protocol, and EDORL and the online+offline baseline use the same expert dataset $\mathcal{D}_{\mathrm{e}}$ (see Table~\ref{table:model_architecture} in Appendix~\ref{app:implement_details}).
\end{revision}

\begin{figure}[h!]
    \centering
    
    \subfigure[\small Linearized Inverted Pendulum\label{fig:policy_evaluation_1}]{%
        \includegraphics[width=0.4\linewidth]{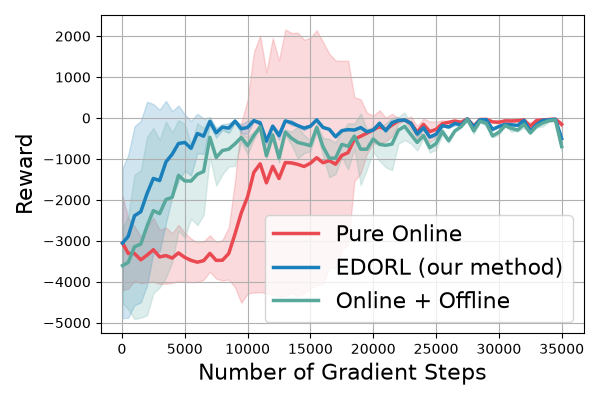}%
    }\hfill
    \subfigure[\small Inverted Pendulum\label{fig:policy_evaluation_2}]{%
        \includegraphics[width=0.4\linewidth]{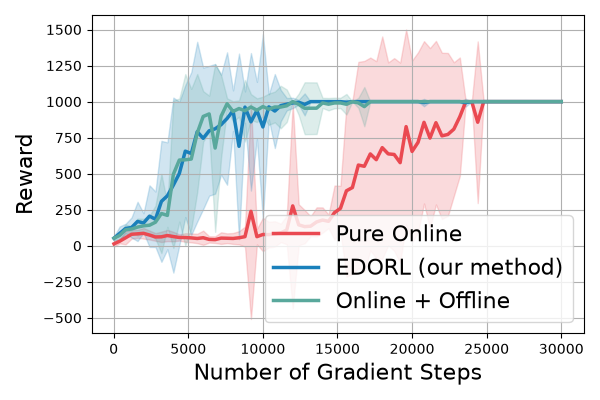}%
    }\hfill
    
    \subfigure[\small Inverted Double Pendulum\label{fig:policy_evaluation_3}]{%
        \includegraphics[width=0.4\linewidth]{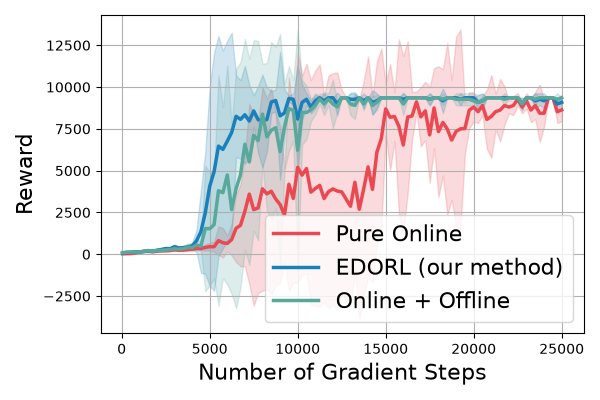}%
    }\hfill
    \subfigure[\small Hopper\label{fig:policy_evaluation_4}]{%
        \includegraphics[width=0.4\linewidth]{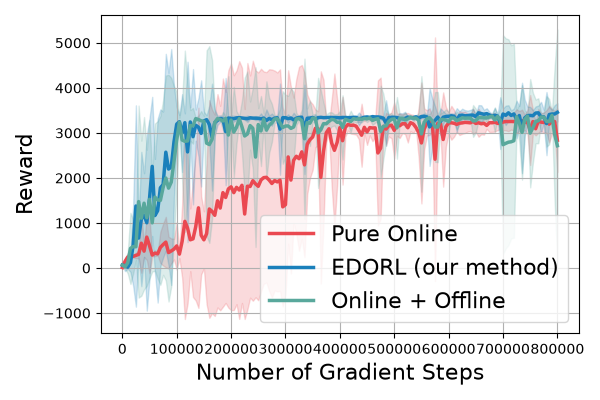}%
    }
    \caption{\rev{Evaluation reward for EDORL and the baselines. Shaded bands show mean $\pm 1.96$ standard deviations across five seeds, with 20 episodes per evaluation.}}
    \label{fig:policy_evaluation}
\end{figure}

\subsection{Results}
\label{subsec:result}

\rev{Figure~\ref{fig:policy_evaluation} shows that EDORL reaches the expert performance level in fewer gradient steps than pure online training on all evaluated environments.} On a complex system, Hopper, EDORL reaches the expert level after $2\times10^5$ gradient steps, while pure online training needs around $7\times10^5$, \rev{a reduction of approximately $3.5\times$ in gradient steps.} We present the amount of speed-up and reduction of sample requirement during training for all the environments in Table~\ref{tab:learning_speedup}. \rev{This is consistent with the role of} the behaviour cloning regulariser computed on $\mathcal{D}_{\mathrm{e}}$, which keeps the policy near the expert in the early phase, when the critic is still poorly estimated and the online data is of low quality.\\

The speed-up of EDORL also depends on the quality of the expert data. When the expert is far from optimal, the BC term pulls the policy toward a worse solution, and the gain shrinks. We study this effect by varying the expert quality and plot the learning curves in Fig.~\ref{fig:expert_quality}. \rev{On \texttt{InvertedPendulum}, lower-quality expert data slows learning compared with higher-quality expert data. A medium-quality expert on \texttt{InvertedDoublePendulum} shows the same trend.}

\begin{table}[H]
\centering
\caption{\rev{Gradient steps to reach the expert performance level. The reduction factor is the online RL step count divided by the EDORL step count.}}
\label{tab:learning_speedup}
\vskip 0.15in
\begin{tabular}{lccc}
\toprule
 & \textbf{$\#$ of gradient steps} & \textbf{$\#$ of gradient steps} & \rev{\textbf{$\#$ of sample reduction}} \\
\textbf{Environment} & \textbf{for pure online} & \textbf{for EDORL} & \rev{\textbf{for EDORL (factor)}} \\
\midrule
Linearized Pendulum  & $23,278$        & $\mathbf{11,352}$ & $2.05 $ \\
Inverted Pendulum     & $24,508$        & $\mathbf{13,417}$ & $1.83 $ \\
Double Pendulum       & $24,047$         & $\mathbf{12,761}$ & $1.88 $ \\
Hopper             & $7 \times 10^5$      & $\mathbf{2\times10^5}$ & $3.5 $ \\
\bottomrule
\end{tabular}
\end{table}

\begin{figure}[h!]
    \centering
    
    \subfigure[\small Inverted Pendulum \label{fig:expert_quality_1}]{%
        \includegraphics[width=0.4\linewidth]{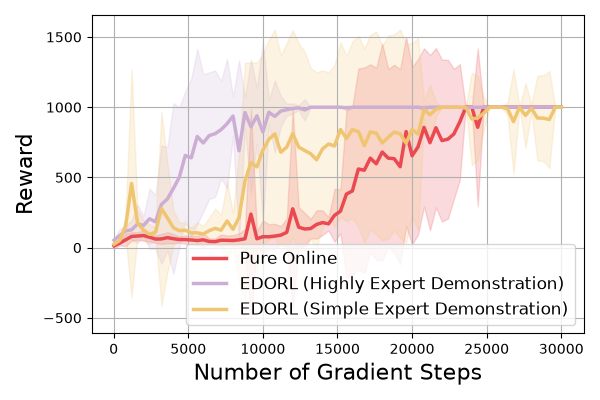}%
    }\hfill
    \subfigure[\small Inverted Double Pendulum\label{fig:expert_quality_2}]{%
        \includegraphics[width=0.4\linewidth]{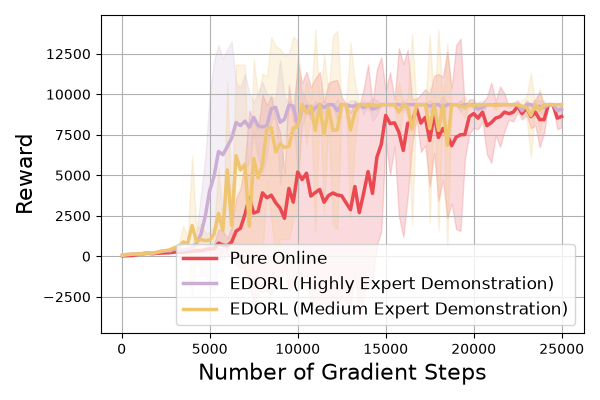}%
    }
    \caption{\rev{Effect of expert data quality on EDORL learning curves for Inverted Pendulum and Inverted Double Pendulum.}}
    \label{fig:expert_quality}
\end{figure}

%% file: sections/conclusion.tex
\section{Conclusion}
\label{sec:conclusion}
In this paper, we have studied expert driven online RL (EDORL), and proposed to use the average expected cost regularised with the behavioural cloning loss from the expert data to learn the policy suitable for closed-loop control. Formulating the policy learning as a policy gradient problem allowed us to theoretically study the convergence and stabilisability of the policy learning process in the context of LQR problems in the presence of the expert data. We showed that keeping the closed-loop average cost along with BC regularisation, preserves coercivity of the loss and thereby guarantees that every iterate remains stabilising when the learning rate is appropriately selected. \rev{Under the stated assumptions, gradient dominance yields linear convergence of exact policy gradient updates.} \rev{Experiments on linear and nonlinear systems show that EDORL reaches comparable performance in fewer gradient steps than pure online RL.}  \rev{Future work will extend the convergence analysis to the complete actor--critic algorithm, including the coupled updates of the critic and policy.}

%% file: sections/appendix.tex
\input{sections/apex_algorithm}
\input{sections/optimization_preliminaries}
\input{sections/proof_lem_sigma}
\input{sections/data_buffer}
\input{sections/proof_of_coerciveness}
\input{sections/proof_lem_gradient}
\input{sections/proof_theo_gradient_dominant}
\input{sections/proof_theo_stab}
\input{sections/apex_experimental_analysis}

%% file: sections/apex_algorithm.tex
\section{EDORL Algorithm}
\label{app:edorl_algo}
\RestyleAlgo{ruled}
\begin{algorithm2e}[h!]
\begin{revision}
\caption{\rev{EDORL: Expert-Driven Online RL}}
\label{alg:edorl}
\SetKwProg{Fn}{Function}{:}{}
\Fn{$\mu_{\theta_b}(x,\eta)$}{
    \For{$t=0,\dots,M-1$}{
        $\eta \gets \eta + v_{\theta_b}(t/M,x,\eta)/M$\;
    }
    \Return $\eta$\;
}
\BlankLine
\textbf{Step 1: Train the expert policy.}
Train $\pi_{\mathrm{e}}$ as described in Section~\ref{subsec:expert_data_creation}\;
\textbf{Step 2: Collect expert data.}
Generate transitions using $\pi_{\mathrm{e}}$ and store them in $\mathcal{D}_{\mathrm{e}}$\;
\BlankLine
\textbf{Step 3: Train the flow policy.}\;
\While{not converged}{
    Sample a minibatch $(x,u)\sim\mathcal{D}_{\mathrm{o}}$\;
    $a^0\sim\mathcal{N}(0,I_{n_u})$, $a^1\gets u$, $t\sim\mathrm{Unif}[0,1]$\;
    $a^t\gets(1-t)a^0+ta^1$\;
    Update $\theta_b$ to minimize
    $\mathbb{E}[\|v_{\theta_b}(t,x,a^t)-(a^1-a^0)\|_2^2]$, following \eqref{eq:flow_matching}\;
}
\BlankLine
\textbf{Step 4: Actor--critic policy learning.}\;
\While{not converged}{
    Receive the initial state $x_0$\;
    \For{$t=0,\dots,T$}{
        Take action $a_t\sim\pi_{\varphi}(\cdot\mid x_t)$\;
        Store $(x_t,a_t,c_t,x_{t+1})$ in $\mathcal{D}_{\mathrm{o}}$\;
        \For{$g=1,\dots,G$}{
            Sample minibatch $B_{\mathrm{o}}$ from $\mathcal{D}_{\mathrm{o}}$\;
            \textit{Train the critic using only $B_{\mathrm{o}}$:}\;
            $\eta\sim\mathcal{N}(0,I_{n_u})$, $u'\gets\pi_{\varphi}(x',\eta)$\;
            Update $\theta_q$ by TD learning \citep{tesauro1995temporal}\;
            Sample minibatch $B_{\mathrm{e}}$ from $\mathcal{D}_{\mathrm{e}}$\;
            \textit{Train the actor using both $B_{\mathrm{o}}$ and $B_{\mathrm{e}}$:}\;
            Compute $J(\varphi)$ from \corr{\eqref{eq:average_cost}} using $B_{\mathrm{o}}$\;
            Compute $\mathcal{R}(\varphi)$ from \eqref{eq:bc_regularizer} using $B_{\mathrm{e}}$ and $\mu_{\theta_b}$\;
            Update the actor parameters $\varphi$ using the gradients\;
        }
    }
}
\end{revision}
\end{algorithm2e}
\FloatBarrier

%% file: sections/optimization_preliminaries.tex
\begin{revision}
\section{Optimisation Preliminaries}\label{app:optimization_definitions}
\noindent
\textbf{Coercive function:} A function $J:\mathcal{K}\to\mathbb{R}$ is said to be coercive on $\mathcal{K}$ if, for every sequence $\{K^{t}\}_{t=1}^{\infty}\subseteq \mathcal{K}$,
{\setlength{\belowdisplayskip}{2mm} \setlength{\abovedisplayskip}{1mm}
\begin{align*}
    J(K^{t}) \to \infty
    \qquad \text{whenever} \qquad
    \|K^{t}\|_{2}\to\infty
    \ \text{or}\
    K^{t}\to \partial \mathcal{K}.
\end{align*}
}
Here, $\partial \mathcal{K}$ denotes the boundary of the feasible set $\mathcal{K}$. In particular, when $\mathcal{K}$ is the set of stabilising controller gains, $\partial \mathcal{K}$ corresponds to the boundary between stability and instability.

In other words, the function value becomes unbounded whenever the sequence either grows without bound or approaches the boundary of $\mathcal{K}$. Equivalently, all sublevel sets of $J$ are compact; that is, for any $\gamma\in\mathbb{R}$, the set $\mathcal{K}_{\gamma}:=\{K\in\mathcal{K}\,:\, J(K)\leq \gamma\}$ is compact.

\vspace{5pt}
\noindent
\textbf{L-smooth function:} A function $J: \mathcal{K} \rightarrow \mathbb{R},\: \mathcal{K} \subseteq \mathbb{R}^n$ is L-smooth if for all $K_1,\:K_2 \in \mathcal{K}$
{\setlength{\belowdisplayskip}{2mm} \setlength{\abovedisplayskip}{2mm}
\begin{align*}
    \|  \nabla J(K_1) - \nabla J(K_2) \|_{2} \leq L \| K_1-K_2 \|_{2},
\end{align*}
}
where $L>0$ is called the smoothness parameter. \corr{On an open convex domain, a twice continuously differentiable function is $L$-smooth if and only if $\|\nabla^2J(K)\|_2\leq L$ everywhere, equivalently $-LI\preceq\nabla^2J(K)\preceq LI$. A Hessian norm bound also implies smoothness on any convex subset by integration along segments.}

\vspace{5pt}
\noindent
\textbf{Gradient dominant function:} A function $J: \mathcal{K} \rightarrow \mathbb{R}$ is gradient dominant if there exists a constant $\mu >0$ such that for all $K \in \mathcal{K}$
{\setlength{\belowdisplayskip}{2mm} \setlength{\abovedisplayskip}{2mm}
\begin{align*}
    J(K)-J(K^{*}) \leq \frac{1}{2\mu} \| \nabla J(K) \|_{F}^{2},
\end{align*}
}
where $K^{*}$ is a global minimiser of $J$.
\end{revision}

%% file: sections/proof_lem_sigma.tex
\section{\rev{Stationary Covariance and Average Cost}}\label{app:proof_Lem_sigma}
\begin{Lem}
\corr{Let Assumption~\ref{Assum:lqr_regularity} hold and let $K\in\mathcal{K}$.} Then, the covariance of the state variable $x_k$ under the policy $\pi_{K} (x_k,z_k) = K x_k+z_k$ converges to a steady state value $\Sigma_K$ which is the unique positive definite solution of the Lyapunov equation
{\setlength{\belowdisplayskip}{2mm} \setlength{\abovedisplayskip}{2mm}
\begin{align}
    \Sigma_K=(A+BK)\Sigma_K(A+BK)^{\top}+B W_{z} B^{\top} +W_{w}.
    \label{eq:Sigma_steady}
\end{align}
}
\label{Lem:sigma}
\end{Lem}
\begin{correction}
\begin{proof}
Put $F_K=A+BK$ and let $m_0$ and $\Sigma_0$ be the initial mean and covariance. Independence of the noise gives
\begin{align*}
m_k&=F_K^k m_0,\\
\Sigma_{k+1}&=F_K\Sigma_kF_K^{\top}+\Omega,\\
\Sigma_k&=F_K^k\Sigma_0(F_K^k)^{\top}
+\sum_{j=0}^{k-1}F_K^j\Omega(F_K^j)^{\top}.
\end{align*}
Since $\rho(F_K)<1$, the initial-covariance term vanishes and the series converges to $\Sigma_K\succeq\Omega\succ0$. It solves \eqref{eq:Sigma_steady}. If two solutions differ by $D$, then $D=F_K^jD(F_K^j)^{\top}$ for every $j$; taking $j\to\infty$ shows uniqueness. Moreover $m_k\to0$, so $\mathbb{E}[x_kx_k^{\top}]\to\Sigma_K$. This argument needs only a finite initial second moment, not a Gaussian initial state.
\end{proof}
\end{correction}

\begin{Lem}
    Consider the linear system in \eqref{eq:lin:dynamics}, average expected cost in \eqref{Eq:Quad:cost}-\eqref{eq:Ave:Quad:cost} \corr{under Assumption~\ref{Assum:lqr_regularity} with $K\in\mathcal{K}$.} The average expected cost associated with the policy $\pi_{K} (x,z)= K x+z, \quad z \sim \mathcal{N}(0,W_z)$, is given by
    {\setlength{\belowdisplayskip}{2mm} \setlength{\abovedisplayskip}{2mm}
    \begin{align}
        J(K)&= \tfrac{1}{2}\Tr ((R_x + K^{\top} R_u K)\Sigma_K) + \tfrac{1}{2}\Tr(R_u W_z)\notag\\
        &= \tfrac{1}{2}\Tr(P_K (BW_z B^{\top} +W_w)) + \tfrac{1}{2}\Tr(R_u W_z),
	\label{eq:Ave:One_step_policy}
\end{align}
}
where $\Sigma_K$ is given in \eqref{eq:Sigma_steady} and $P_K$ is the unique positive semi-definite solution to the following Bellman (Lyapunov) equation
{\setlength{\belowdisplayskip}{2mm} \setlength{\abovedisplayskip}{2mm}
\begin{align}
    P_K = (R_x+K^{\top} R_u K) + (A+BK)^{\top} P_K (A+BK).
    \label{eq:Pk_Bellman}
\end{align}
}
\label{Lem:avecost}
\end{Lem}
\begin{proof}
We derive the average cost by directly using the policy $\pi_{K} (x,z)= K x+z$ in the average cost \eqref{eq:Ave:Quad:cost}
{\setlength{\belowdisplayskip}{2mm} \setlength{\abovedisplayskip}{2mm}
\begin{align*}
	J(K) =\lim_{T \rightarrow \infty}& \frac{1}{2T} \mathbb{E}^{\pi_K}_{x_0\sim\nu_0} \sum_{k=0}^{T-1}\Bigl[ x_{k}^{\top} R_x x_k +(Kx_{k}+z_k)^{\top} R_u (Kx_k+z_k)\Bigr]\\
    =\lim_{T \rightarrow \infty}& \frac{1}{2T} \mathbb{E}^{\pi_K}_{x_0\sim\nu_0}\sum_{k=0}^{T-1}\Bigl[x_{k}^{\top} R_x x_k+2z_{k}^{\top} R_u K x_k +x_k^{\top} K^{\top} R_u K x_k+z_k^{\top}  R_u z_k\Bigr]
\end{align*}
}
Using $ \mathbb{E}^{\pi_K}[z_k^{\top} R_u K x_k]=0$ (which holds since $z_k$ is zero mean and independent of $x_k$) and $\Tr(AB)=\Tr(BA)$, one gets
{\setlength{\belowdisplayskip}{2mm} \setlength{\abovedisplayskip}{2mm}
\begin{align*}
	J(K)=&\lim_{T \rightarrow \infty} \frac{1}{2T} \sum_{k=0}^{T-1}\Bigl[ \Tr \bigl((R_x+K^{\top} R_u K)  \mathbb{E}^{\pi_K}[x_kx_{k}^{\top}]\bigr) + \Tr(R_u \mathbb{E}_{z \sim \mathcal{N}(0,W_z)}[z_k z_k^{\top}])\Bigr].
\end{align*}
}
Since $K \in \mathcal{K}$, Lemma \ref{Lem:sigma} gives $\mathbb{E}^{\pi_K}[x_{k}x_{k}^{\top}]\to\Sigma_K$ as $k\to\infty$, where $\Sigma_K$ solves \eqref{eq:Sigma_steady}. The Ces\`aro average of a convergent sequence converges to the same limit, so the effect of $\nu_0$ vanishes and
\begin{align*}
	J(K)= \tfrac{1}{2}\Tr ((R_x+K^{\top} R_u K)  \Sigma_K) + \tfrac{1}{2}\Tr(R_u W_z).
\end{align*}
Using \eqref{eq:Pk_Bellman} to substitute $R_x+K^{\top} R_u K = P_K-(A+BK)^{\top}P_K(A+BK)$ and then \eqref{eq:Sigma_steady}, we obtain
\begin{align*}
\Tr ((R_x+K^{\top} R_u K)  \Sigma_K)= \Tr(P_K (BW_zB^{\top}+W_w)),
\end{align*}
which gives the second equality in \eqref{eq:Ave:One_step_policy}.
\end{proof}

%% file: sections/data_buffer.tex
\section{The role of the data buffers}\label{app:data_buffer}
We first make precise the role of two data buffers $\mathcal{D}_{\mathrm{o}}$ (online data buffer) and $\mathcal{D}_{\mathrm{e}}$ (expert data buffer). Since the two terms of the EDORL policy loss (Eq.~\eqref{eq:edorl_lqr}) are defined with respect to two \emph{different} distributions, conflating them would invalidate the analysis.

\rev{The expert data assumption and fixed moments are defined in Section~\ref{subsec:lqr_policy_loss}.}

The two terms of the EDORL policy loss are then treated as follows.
\begin{itemize}
    \item The performance term is the \emph{true closed-loop} average cost $J(K)$ of \eqref{eq:Ave:Quad:cost}. By Lemma \ref{Lem:avecost} it is a deterministic function of $K$ through the policy-dependent matrices $\Sigma_K$ and $P_K$. It involves \emph{no} dataset distribution.
    \item The behavioural cloning term is a \emph{supervised regression} onto the expert demonstrations. Its defining distribution \emph{is} the expert state distribution. 
\end{itemize}
Under Assumption \ref{Assum:data}, the resulting objective is a deterministic function of $K$ alone, and both the coercivity of $J$ and the constancy of the regulariser's curvature are preserved.

%% file: sections/proof_of_coerciveness.tex
\section{\corr{Proof of Lemma~\ref{lem:coercive}}}
\label{app:proof_lem_coercive}
\begin{correction}
\begin{proof}
Let $F_K=A+BK$, $q_x=\lambda_{\min}(R_x)>0$, $q_u=\lambda_{\min}(R_u)>0$, and $\omega=\lambda_{\min}(\Omega)>0$. For $K\in\mathcal{K}$, the stationary covariance has the convergent series
\begin{align*}
\Sigma_K=\sum_{j=0}^{\infty}F_K^j\Omega(F_K^j)^{\top}\succeq\Omega\succeq\omega I_n.
\end{align*}
The average-cost formula therefore implies both
\begin{align*}
J(K)&\geq\frac{q_u\omega}{2}\|K\|_F^2,\\
J(K)&\geq\frac{q_x}{2}\Tr(\Sigma_K)
\geq\frac{q_x\omega}{2}\sum_{j=0}^{\infty}\|F_K^j\|_F^2
\geq\frac{q_x\omega}{2(1-\rho(F_K)^2)}.
\end{align*}
For the last inequality, $\|F_K^j\|_F\geq\rho(F_K^j)=\rho(F_K)^j$. The first bound diverges as $\|K\|_F\to\infty$ and the second diverges as a stabilising gain approaches a finite stability-boundary point. Continuity of the objective on $\mathcal{K}$ then makes it closed in the ambient finite-dimensional space, hence compact. The positivity of $R_u$ alone would not imply $R_x+K^{\top}R_uK\succ0$; the explicit regularity assumptions are used here.
\end{proof}
\end{correction}

%% file: sections/proof_lem_gradient.tex
\section{Proof of Lemma~\ref{Lem:gradient}}\label{app:proof_Lem_gradient}
\begin{revision}
\subsection{Q-Function and Policy Structure}
\label{subsec:q_function}

For the linear system in \eqref{eq:lin:dynamics} with the quadratic cost in \eqref{Eq:Quad:cost} and the average cost formulation in \eqref{eq:Ave:Quad:cost}, the relative $Q$-function of the policy $\pi_K$ is given by \rev{\citep{adib2023linear}}
{\setlength{\belowdisplayskip}{2pt} \setlength{\abovedisplayskip}{5pt}
\begin{align}
    Q^{K}(x,u) =\frac{1}{2} \begin{bmatrix}
        x\\
        u
    \end{bmatrix}^{\top} \begin{bmatrix}
        S_{xx}(K) & S_{ux}^{\top}(K)\\
        S_{ux}(K) & S_{uu}(K)
    \end{bmatrix}\begin{bmatrix}
        x\\
        u
    \end{bmatrix}-J(K), 
    \label{Eq:q_function}
\end{align}
}
where the blocks are given explicitly in terms of the system parameters $\theta=(A,B)$ \emph{and} the policy gain $K$ by
{\setlength{\belowdisplayskip}{2pt} \setlength{\abovedisplayskip}{2pt}
\begin{equation}
    S_{xx}(K) = R_x + A^{\top} P_K A, \quad S_{ux}(K) = B^{\top} P_K A, \quad S_{uu}(K) = R_u + B^{\top} P_K B. \label{eq:S_blocks}
\end{equation}

}
with $P_K$ the solution of the policy-dependent Bellman equation \eqref{eq:Pk_Bellman}.

\begin{Lem} [Gradient of the EDORL policy loss]
    Consider the linear system in \eqref{eq:lin:dynamics}, average expected cost in \eqref{Eq:Quad:cost}-\eqref{eq:Ave:Quad:cost}, \corr{let Assumptions~\ref{Assum:lqr_regularity} and~\ref{Assum:data} hold, and let $K\in\mathcal{K}$.} The gradient of the one-step policy loss $ l_{\varphi}^{\text{lin}}$ \eqref{eq:loss_one_step_LQ} is given by
{\setlength{\belowdisplayskip}{2pt} \setlength{\abovedisplayskip}{6pt}
\begin{align}
\begin{split}
     \nabla_K l_{\varphi}^{\text{lin}}(K)=&\underbrace{\left( S_{uu}(K) K + S_{ux}(K)\right)\Sigma_K}_{=\nabla_K J(K)}+ \alpha \left(K G_{\mathcal{D}_{\mathrm{e}}} - M\right),
\end{split}
   \label{eq:nabla_loss_2}
\end{align}
}
where $S_{uu}(K)=R_u+B^{\top}P_KB$ and $S_{ux}(K)=B^{\top}P_KA$ are given in \eqref{eq:S_blocks}, $P_K$ solves \eqref{eq:Pk_Bellman}, $\Sigma_K$ solves \eqref{eq:Sigma_steady}, and $G_{\mathcal{D}_{\mathrm{e}}}$, $M$ are the constant data moments \eqref{eq:GD}, \eqref{eq:M}.
\label{Lem:gradient}
\end{Lem}
\end{revision}

\begin{correction}
\begin{proof}
Write $F_K=A+BK$ and let $H$ be an arbitrary gain perturbation. The Lyapunov solutions are differentiable on $\mathcal{K}$. Differentiating the Bellman equation in direction $H$, with $\dot P=\mathrm{D}P_K[H]$, yields
\begin{align*}
\dot P-F_K^{\top}\dot P F_K&=H^{\top}E_K+E_K^{\top}H,\\
E_K&:=R_uK+B^{\top}P_KF_K=S_{uu}(K)K+S_{ux}(K).
\end{align*}
Since $\Omega$ and $W_z$ are fixed and $\Omega=\Sigma_K-F_K\Sigma_KF_K^{\top}$,
\begin{align*}
\mathrm{D}J(K)[H]
&=\tfrac12\Tr(\dot P\Omega)
=\tfrac12\Tr((\dot P-F_K^{\top}\dot P F_K)\Sigma_K)\\
&=\langle E_K\Sigma_K,H\rangle_F.
\end{align*}
Thus $\nabla_KJ(K)=E_K\Sigma_K$, including the factor $1/2$ in the stage cost. For the cloning term, the fixed behaviour policy and Assumption~\ref{Assum:data} give
\begin{align*}
\nabla_KL_{\mathrm{Distill}}(K)
&=\alpha\mathbb{E}_{x,z}[(Kx+z-\pi_{\varphi_b}(x,z))x^{\top}]\\
&=\alpha(KG_{\mathcal{D}_{\mathrm{e}}}-M),
\end{align*}
because $\mathbb{E}[zx^{\top}]=0$. Adding these expressions proves \eqref{eq:nabla_loss_2}. The proof does not require $W_z$ to be invertible.
\end{proof}
\end{correction}

%% file: sections/proof_theo_gradient_dominant.tex
\section{\rev{Proof of Theorem \ref{Theo:EDORL_grad_dominence}}}
\label{app:proof_Theo_gradient_dominant}
\begin{correction}
\begin{proof}
Write $f(K)=l_{\varphi}^{\mathrm{lin}}(K)$ and $G=G_{\mathcal{D}_{\mathrm{e}}}$. Since the behaviour policy is fixed and $\mathbb{E}[zx^{\top}]=0$, the cloning term expands as
\begin{align*}
L_{\mathrm{Distill}}(K)=\frac{\alpha}{2}\Tr(KGK^{\top})-\alpha\Tr(KM^{\top})+C_b,
\end{align*}
where $C_b$ is finite and independent of $K$. Consequently, for every $K\in\mathcal{C}$ and direction $V\in\mathbb{R}^{m\times n}$,
\begin{align*}
\mathrm{vec}(V)^{\top}\nabla_K^2 f(K)\mathrm{vec}(V)
&=\mathrm{vec}(V)^{\top}\nabla_K^2J(K)\mathrm{vec}(V)
+\alpha\Tr(VGV^{\top})\\
&\geq\bigl(\alpha\sigma-c_J(\mathcal{C})\bigr)\|V\|_F^2
=\mu\|V\|_F^2.
\end{align*}
The bound holds on all of $\mathcal{C}$, so integration along any segment in this convex set proves $\mu$-strong convexity there.

For existence, $\mathcal{K}_{\gamma_0}$ is nonempty and compact, and $f$ is continuous. It therefore attains a minimum on this sublevel set. Every point outside the set has value greater than $\gamma_0$, so this minimum is global over $\mathcal{K}$. Every global minimiser belongs to $\mathcal{K}_{\gamma_0}\subset\mathcal{C}$; strong convexity on $\mathcal{C}$ makes it unique. The $r$-neighbourhood condition places $K_\alpha^*$ in $\operatorname{int}\mathcal{C}$, and $\nabla f(K_\alpha^*)=0$ because it is a minimiser on the open set $\mathcal{K}$.

For any $K\in\mathcal{C}$, strong convexity applied to the pair $(K,K_\alpha^*)$ gives
\begin{align*}
f(K)-f(K_\alpha^*)
&\leq\langle\nabla f(K),K-K_\alpha^*\rangle-\frac{\mu}{2}\|K-K_\alpha^*\|_F^2\\
&\leq\frac{\|\nabla f(K)\|_F^2}{2\mu},
\end{align*}
where the last inequality follows by completing the square. This proves \eqref{eq:grad_dominant} without assuming that $K-\mu^{-1}\nabla f(K)$ belongs to $\mathcal{C}$. Applying strong convexity at $K_\alpha^*$ also yields the quadratic-growth bound
\begin{align*}
f(K)-f(K_\alpha^*)\geq\frac{\mu}{2}\|K-K_\alpha^*\|_F^2,\qquad K\in\mathcal{C}.
\end{align*}
\end{proof}
\end{correction}

%% file: sections/proof_theo_stab.tex
\section{Proof of Theorem~\ref{Theo:EDORL_Stab}}
\label{app:proof_Theo_Stab}
\begin{correction}
\begin{Lem}[$L$-smoothness on the working region]
\label{lem:l_smoothness}\label{Lem:Lsmooth}
Under the hypotheses of Theorem~\ref{Theo:EDORL_grad_dominence}, the objective is $L$-smooth on $\mathcal{C}$ in the Frobenius norm, with the valid choice
\begin{align}
L:=L_J(\mathcal{C})+\alpha\|G_{\mathcal{D}_{\mathrm{e}}}\|_2,
\quad L_J(\mathcal{C}):=\max_{K\in\mathcal{C}}\|\nabla_K^2J(K)\|_2<\infty.
\label{eq:Lsmooth}
\end{align}
In particular $0<\mu\leq L$.
\end{Lem}
\begin{proof}
The Lyapunov operators defining $P_K$ and $\Sigma_K$ are invertible for every $K\in\mathcal{K}$. Their coefficients depend polynomially on $K$, so their solutions, and hence $J$, are analytic on $\mathcal{K}$. Since $\mathcal{C}$ is compact and contained in $\mathcal{K}$, the continuous Hessian norm has a finite maximum there. The cloning Hessian is $\alpha(G_{\mathcal{D}_{\mathrm{e}}}\otimes I_m)$, so $\|\nabla_K^2 f(K)\|_2\leq L$ for $f=l_{\varphi}^{\mathrm{lin}}$ and $K\in\mathcal{C}$. Integrating the Hessian along segments in the convex set $\mathcal{C}$ gives
\begin{align*}
\|\nabla f(K_1)-\nabla f(K_2)\|_F\leq L\|K_1-K_2\|_F.
\end{align*}
Finally, $0<\mu=\alpha\sigma-c_J(\mathcal{C})\leq\alpha\|G_{\mathcal{D}_{\mathrm{e}}}\|_2\leq L$.
\end{proof}

\begin{proof}[Proof of Theorem~\ref{Theo:EDORL_Stab}]
Write $f=l_{\varphi}^{\mathrm{lin}}$ and $f^*=f(K_\alpha^*)$. Compactness of $\mathcal{K}_{\gamma_0}$ and continuity of $\nabla f$ imply that $\Gamma_0$ is finite. Initially $K^0\in\mathcal{K}_{\gamma_0}$. Suppose inductively that $K^t\in\mathcal{K}_{\gamma_0}$.

\emph{Step 1 (the entire update segment is stabilizing).}
The gradient update satisfies
\begin{align*}
\|K^{t+1}-K^t\|_F=\eta\|\nabla f(K^t)\|_F\leq\eta\Gamma_0\leq r.
\end{align*}
Thus $K^{t+1}\in\mathcal{C}$ by the neighbourhood condition. Convexity places the entire update segment in $\mathcal{C}\subset\mathcal{K}$. This establishes stability of the candidate update before evaluating its objective.

\emph{Step 2 (descent and invariance).}
Lemma~\ref{Lem:Lsmooth} applies along that segment, giving
\begin{align*}
f(K^{t+1})
&\leq f(K^t)-\eta\bigl(1-\eta L/2\bigr)\|\nabla f(K^t)\|_F^2\\
&\leq f(K^t)-\frac{\eta}{2}\|\nabla f(K^t)\|_F^2\leq\gamma_0.
\end{align*}
Consequently $K^{t+1}\in\mathcal{K}_{\gamma_0}$, closing the induction and proving (i)--(ii). Coercivity provides a compact sublevel set separated from the stability boundary; the step-size bound is also essential, since coercivity alone does not prevent a finite update from jumping outside $\mathcal{K}$.

\emph{Step 3 (linear rate and convergence of the gains).}
Theorem~\ref{Theo:EDORL_grad_dominence} gives $\|\nabla f(K^t)\|_F^2\geq2\mu(f(K^t)-f^*)$. Substituting into the descent inequality yields
\begin{align*}
f(K^{t+1})-f^*\leq(1-\eta\mu)(f(K^t)-f^*).
\end{align*}
Since $0<\mu\leq L$ and $0<\eta\leq1/L$, the factor $q=1-\eta\mu$ lies in $[0,1)$. Iterating proves \eqref{eq:linear_rate}. The quadratic-growth bound from the proof of Theorem~\ref{Theo:EDORL_grad_dominence} then gives
\begin{align*}
\|K^t-K_\alpha^*\|_F^2\leq\frac{2}{\mu}q^t(\gamma_0-f^*),
\end{align*}
so $K^t\to K_\alpha^*$. If $\Gamma_0=0$, then $\nabla f(K^0)=0$ and gradient dominance implies $K^0=K_\alpha^*$; the iterates are constant and the conclusions still hold.
\end{proof}
\end{correction}

%% file: sections/apex_experimental_analysis.tex
\section{\rev{Experimental Details}}
\label{app:experimental_analysis}
\subsection{Simulation on Linear System (LQR setup)}
\label{app:linear_system_simulation}
For the simulation studies in LQR setting, we considered a linearized inverted pendulum system as a benchmark example. Here the states are the angle and angular velocity of the pendulum, i.e., $x_k = [\theta_k,\ \dot{\theta}_k]^\top$, and the angular velocity is clipped to $|\dot{\theta}_k| \leq 8\,\text{rad/s}$. The control input $u_k$ corresponds to the applied torque and is bounded by $|u_k| \leq 2\,\text{Nm}$. The noise variance is taken as $W_w = 0.0001 I$. The system is linearized around the upright position. For the cost evaluation, we use $R_x = \text{diag}(1, 0.1)$ and $R_u = 0.001$. Furthermore, the system matrices are given by      
\begin{align}
    A = 
    \begin{bmatrix}
        1 & \Delta t \\
        -\tfrac{g}{l}\Delta t & 1 - \tfrac{b}{m l^2}\Delta t
    \end{bmatrix}, 
    \qquad
    B = 
    \begin{bmatrix}
        0 \\
        \tfrac{\Delta t}{m l^2}
    \end{bmatrix},
\end{align}
where $m$ is the mass, $l$ the rod length, $b$ the viscous damping coefficient, $g$ the gravitational constant, and $\Delta t$ the sampling interval. We have used $m = 1$\,kg, $l = 1$\,m, $b = 0$, $g = 10$\,m/s$^2$, and $\Delta t = 0.05$\,s for the simulation.

To generate the expert data, we train a TD3-like architecture \rev{\citep{fujimoto2018addressing}} using online RL method and then use that trained model to generate expert data $\mathcal{D}_{\mathrm{e}}$. Implementation details of the expert model and the EDORL model are \rev{given} in Appendix~\ref{app:implement_details}.

\begin{revision}
\subsection{Expert Data Creation}
\label{subsec:expert_data_creation}
The expert dataset $\mathcal{D}_{\mathrm{e}} = \bigl\{(x_i, u_i, c_i, x_{i+1})\bigr\}_{i=1}^{N}$ is collected once, before training, by running an expert policy $\pi_{\mathrm{e}}$ in the environment and recording the transitions  at each time step. For the nonlinear dynamical systems from MuJoCo, we obtain $\pi_{\mathrm{e}}$ by training a TD3-like agent \rev{\citep{fujimoto2018addressing}} with online RL, and then roll it out to collect $N$ transitions. We use a different network architecture for the expert ($\pi_{\mathrm{e}}$ ) than for EDORL, so that the learned policy does not benefit from sharing a structure with the data-generating policy. For the expert model architecture, see Table~\ref{table:expert_model_details} in Appendix~\ref{app:implement_details}
\end{revision}

\subsection{Implementation Details}
\label{app:implement_details}
In this section, we present the model architectures and \rev{hyperparameters} that we use for our experiments. And for both expert model and EDORL model we use TD3-like architecture \rev{\citep{fujimoto2018addressing}} for the actor and critic. For ease of use we use \rev{the following abbreviations for the environment names.}
\begin{itemize}
    \item Linearized Inverted Pendulum: \texttt{lip}, 
    \item Inverted Pendulum: \texttt{ip}, 
    \item Inverted Double Pendulum: \texttt{idp}, 
    \item Hopper: \texttt{hopper}
\end{itemize}
\textbf{Expert Model:} \rev{In Table}~\ref{table:expert_model_details}, we present the model architectures and \rev{hyperparameters} used for expert model of each environment which have been used for expert data generation.
\begin{table}[h!]
    \caption{\rev{Architectures and hyperparameters of the expert models used for data generation.}}
    \label{table:expert_model_details}
    \centering
    \begin{tabular}{llcccc}
        \toprule
        \multicolumn{1}{l}{Model} & \multicolumn{1}{l}{Hyperparameters} & \multicolumn{4}{c}{Value} \\
        \cmidrule(lr){3-6}
        & & lip & ip & idp & hopper \\
        \midrule 
        Critic   & Hidden layers & 2  & 3 & 3 & 3 \\
          & Hidden dimension  & [256, 256]  & [256, 256, 256] & [256, 256, 256]   & [256, 256, 256]  \\
          & Activation  & None    & ReLU  & ReLU & ReLU \\
          & Optimizer  & Adam  & Adam  & Adam   & Adam \\
          & Learning rate  & $3\times10^{-4}$ & $3\times10^{-4}$  & $3\times10^{-4}$  & $3\times10^{-4}$ \\
          & Target update rate& $5\times10^{-3}$  & $5\times10^{-3}$  & $5\times10^{-3}$  & $5\times10^{-3}$ \\
        \midrule 
        Actor & Hidden layers    & 1   & 1   & 1  & 1 \\
          & Hidden dimension  & [256, 256]    & [256, 256, 256]  & [256, 256, 256]   & [256, 256, 256]  \\
          & Activation  & None    & ReLU  & ReLU & ReLU \\
          & Optimizer  & Adam   & Adam   & Adam   & Adam \\
          & Learning rate   & $3\times10^{-4}$  & $3\times10^{-4}$  & $3\times10^{-4}$  & $3\times10^{-4}$ \\
        \bottomrule
    \end{tabular}
\end{table}
\rev{In Table}~\ref{table:model_architecture}, we present the model architectures and \rev{hyperparameters} \rev{that we use} for EDORL for different environments. 

\begin{table}[h!]
    \caption{\rev{Architectures and hyperparameters of the EDORL models.}}
    \label{table:model_architecture}
    \centering
    \begin{tabular}{llcccc}
        \toprule
        \multicolumn{1}{l}{Model} & \multicolumn{1}{l}{Hyperparameters} & \multicolumn{4}{c}{Value} \\
        \cmidrule(lr){3-6}
        & & lip & ip & idp & hopper \\
        \midrule 
        Flow model $v_{\theta_b}$ & Hidden layers   & 2   & 2    & 2   & 2 \\
          & Hidden dimension  & [256, 256]   & [256, 256]   & [256, 256]   & [256, 256] \\
          & Activation  & ReLU    & ReLU  & ReLU & ReLU \\
          & Optimizer  & Adam    & Adam  & Adam & Adam \\
          & Learning rate     & $3\times10^{-4}$  & $3\times10^{-4}$  & $3\times10^{-4}$  & $3\times10^{-4}$ \\
        \midrule 
        Critic   & Hidden layers & 1  & 2    & 2   & 2 \\
          & Hidden dimension  & [256]  & [256, 256]  & [256, 256]   & [256, 256]  \\
          & Activation  & None    & ReLU  & ReLU & ReLU \\
          & Optimizer  & Adam  & Adam  & Adam   & Adam \\
          & Learning rate  & $3\times10^{-4}$ & $3\times10^{-4}$  & $3\times10^{-4}$  & $3\times10^{-4}$ \\
          & Target update rate& $5\times10^{-3}$  & $5\times10^{-3}$  & $5\times10^{-3}$  & $5\times10^{-3}$ \\
        \midrule 
        Actor & Hidden layers    & 1   & 2    & 2   & 2 \\
          & Hidden dimension  & [256]   & [256, 256]  & [256, 256]   & [256, 256]  \\
          & Activation  & None    & ReLU  & ReLU & ReLU \\
          & Optimizer  & Adam   & Adam   & Adam   & Adam \\
          & Learning rate   & $3\times10^{-4}$  & $3\times10^{-4}$  & $3\times10^{-4}$  & $3\times10^{-4}$ \\
        \midrule 
        BC Regularizer   & $\alpha$ & 1.0 & 1.0  & 1.0  & 1.0 \\
        \bottomrule
    \end{tabular}
\end{table}